\documentclass[twoside]{aiml}

\usepackage{aimlmacro}

\def\journal{0} 

\usepackage{amsmath}
\usepackage{amssymb}

\usepackage{stmaryrd}   
\usepackage{paralist}
\usepackage[all]{foreign}
\usepackage{tikz}
\usepackage{subfig}
\usepackage{forest}

\usetikzlibrary{arrows,positioning,calc,backgrounds}

\usepackage{trimclip}

\newcommand{\B}[1]{\Box_{#1}}
\newcommand{\D}[1]{\Diamond_{#1}}

\newcommand{\atom}[1]{#1}

\newcommand{\atomSet}{\mathcal{P}}

\newcommand{\modalitiesSet}{\mathcal{I}}

\newcommand{\Lang}{\mathcal{L}}     

\newcommand{\md}{\mathit{md}}

\newcommand{\depth}[1]{d(#1)}
\newcommand{\bound}[1]{b(#1)}

\newcommand{\bisim}{\leftrightarroweq} 

\newcommand{\contr}[2]{{\lfloor #1 \rfloor}_{#2}}
\newcommand{\rootedContr}[3]{{\llfloor #1 \rrfloor}_{#2}^{#3}}

\newcommand{\class}[2]{[#1]_{#2}}
\newcommand{\reprClass}[1]{\llbracket #1 \rrbracket}

\newcommand{\maxRepr}[1]{#1^{\max}}
\newcommand{\leastRepr}[2]{\min_{#2}(#1)}

\tikzstyle{world}       =[circle,   thick,draw=black,       fill=black,minimum size=4pt,inner sep=0pt]
\tikzstyle{poinetdworld}=[circle,   thick,draw=black,double,fill=black,minimum size=4pt,inner sep=0pt]

\def\lastname{Bolander and Burigana}

\begin{document}
    \begin{frontmatter}
        \title{Better Bounded Bisimulation Contractions (Preprint)}\footnote{This work is currently under review.}
        \author{Thomas Bolander}\footnote{E-mail: \url{tobo@dtu.dk}.}
        \address{Technical University of Denmark, Denmark}
        \author{Alessandro Burigana}\footnote{E-mail: \url{burigana@inf.unibz.it}.}       
        \address{Free University of Bozen-Bolzano, Italy}

        \begin{abstract}
    Bisimulations are standard in modal logic and, more generally, in the theory of state-transition systems. The quotient structure of a Kripke model with respect to the bisimulation relation is called a bisimulation contraction. The bisimulation contraction is a minimal model bisimilar to the original model, and hence, for (image-)finite models, a minimal model modally equivalent to the original. Similar definitions exist for bounded bisimulations ($k$-bisimulations) and bounded bisimulation contractions. Two finite models are $k$-bisimilar if and only if they are modally equivalent up to modal depth $k$. However, the quotient structure with respect to the $k$-bisimulation relation does not guarantee a minimal model preserving modal equivalence to depth $k$. In this paper, we remedy this asymmetry to standard bisimulations and provide a novel definition of bounded contractions called rooted $k$-contractions. We prove that rooted $k$-contractions preserve $k$-bisimilarity and are minimal with this property. Finally, we show that rooted $k$-contractions can be exponentially more succinct than standard $k$-contractions.
    \if\journal1
    Finally, we introduce canonical $k$-contractions, providing unique minimal representatives of every class of $k$-bisimilar models.
    \fi 
\end{abstract}


    \end{frontmatter}

    \section{Introduction}
    Bisimulation plays a central role in countless fields, such as modal logic, set theory, formal verification, concurrency theory, process calculus, and others. Two structures are bisimilar if they are indistinguishable with respect to some behavioral property. In the case of Kripke models, bisimilarity between two models mean that their accessibility relations are structurally equivalent, and it then follows that two (image-)finite models are bisimilar if and only if they are modally equivalent (modal equivalence means they satisfy the same formulas)~\cite{book/cup/Blackburn2001}.   
    When using Kripke models for computational purposes, it is often desirable to keep the models as small as possible while preserving their logical properties, \eg preserving modal equivalence. The quotient structure of a Kripke model with respect to the bisimulation relation is called a \emph{bisimulation contraction}, giving us a minimal model modally equivalent to the original~\cite{blackburn2006modal}. 
    
    Bounded bisimulations only preserve structural equivalence up to some depth $k$. In the case of Kripke models, two models are $k$-bisimilar if their accessibility relations are structurally equivalent up to depth $k$. Two finite models are then $k$-bisimilar if and only if they are modally equivalent to modal depth $k$~\cite{book/cup/Blackburn2001}. When the modalities are used to represent knowledge in an epistemic logic, $k$-bisimilarity means that higher-order reasoning is preserved to depth $k$. If we are only interested in reasoning to depth $k$, say in a setting with agents having bounded rationality, it seems intuitive to consider the quotient structure with respect to $k$-bisimilarity in the attempt to find a minimal model preserving modal equivalence to depth $k$. We call the quotient structure with respect to $k$-bisimilarity the \emph{standard $k$-contraction}~\cite{conf/ijcai/Yu2013}. 
    However, as we will show, the standard $k$-contraction does not guarantee a minimal model preserving modal equivalence to depth $k$. We provide an alternative notion of $k$-bisimulation contraction, called a \emph{rooted $k$-contraction}, that indeed guarantees minimality. 
    
    The inspiration for this paper came from bounded rationality in epistemic planning~\cite{belle2022epistemic}. Epistemic planning is concerned with computing plans for agents having incomplete information about the world and each other's knowledge, using epistemic logic as the underlying formalism~\cite{bolander2011epistemic}.  
     Unfortunately, epistemic planning is in general undecidable, \ie it has an undecidable plan existence problem~\cite{bolander2011epistemic,bolander2020del}. In the search for interesting decidable fragment of epistemic planning, we decided to limit the reasoning capabilities of agents to some fixed depth $k$ (e.g.\ limiting to depth 2 would mean that agents can reason about what they know about the knowledge of others, but not what they know about what others know about them). In order to define such \emph{depth-limited epistemic planning} formally, we needed a notion of contraction that would preserve modal equivalence to depth $k$. We originally started out using standard $k$-contractions, but soon discovered that they didn't guarantee minimality of the contracted models, in many cases actually quite far from it. We then set out on a quest to try to find a better notion of $k$-contraction that would guarantee minimality among $k$-bisimilarity preserving models, hence also reestablishing the symmetry to the corresponding existing results for standard bisimulations. In this paper we report on the results of that quest. We will report on the results of applying these notions to epistemic planning in a separate paper.

\section{Preliminaries}\label{sec:preliminaries}
    In this section, we recall some basic notions in modal logic, \ie pointed Kripke models, bisimulation and bounded bisimulation~\cite{book/cup/Blackburn2001}.
    Let $\atomSet$ be a countable set of atomic propositions and $\modalitiesSet$ a finite set of modality indices.
    The language $\Lang$ of \emph{multi-modal logic} is defined by the following BNF (where $\atom{p} \in \atomSet$ and $i \in \modalitiesSet$): 
    \begin{equation*}
     \varphi ::= \atom{p} \mid \neg \varphi \mid \varphi \wedge \varphi \mid \B{i} \varphi. 
    \end{equation*}
    Symbols $\top$, $\bot$, $\vee$ and $\D{i}$ are defined as usual.
    \emph{Modal depth} is defined inductively on the structure of formulas: $\md(p) = 0$ (for all $p \in \atomSet$), $\md(\neg \varphi) = \md(\varphi)$, $\md(\varphi_1 \land \varphi_2) = \max\{\md(\varphi_1), \md(\varphi_2)\}$ and $\md(\B{i} \varphi) = 1 + \md(\varphi)$.

    \begin{definition}\label{def:model}
        A \emph{model} of $\Lang$ is a triple $M = (W, R, V)$ where:
        \begin{compactitem}
            \item $W \neq \varnothing$ is a finite set of \emph{(possible) worlds};
            \item $R: \modalitiesSet \rightarrow 2^{W \times W}$ assigns to each $i \in \modalitiesSet$ an \emph{accessibility relation} $R_i$;
            \item $V: \atomSet \rightarrow 2^W$ is a \emph{valuation function} assigning to each atom a set of worlds.
        \end{compactitem}
        \noindent A \emph{pointed model} $\mathcal{M}$ is a pair $(M, w_d)$, where $w_d \in W$ is the \emph{designated world}.
    \end{definition}
    We also use $w R_i v$ for $(w, v) \in R_i$. We call an \emph{$i$-edge}, or simply an \emph{edge}, such a pair of worlds. A \emph{path} is a sequence of worlds connected by edges. Note that we have restricted our attention to finite models. All of our results generalize to infinite models, but for many of the results we then have to make additional assumptions such as the models being image-finite, the underlying set of propositional atoms being finite or the number of modalities being finite~\cite{book/cup/Blackburn2001}. To avoid this additional layer of complexity, and since all of our intended applications are within finite models, we restrict to those throughout the paper.


    \begin{definition}\label{def:truth}
        Let $M = (W, R, V)$ be a model of $\Lang$ and let $w \in W$.
        \begin{equation*}
            \begin{array}{@{}lll}                    
                (M, w) \models \atom{p}            & \text{ iff } & w \in V(p) \\
                (M, w) \models \neg \varphi        & \text{ iff } & (M, w) \not\models \varphi \\
                (M, w) \models \varphi \wedge \psi & \text{ iff } & (M, w) \models \varphi \text{ and } (M, w) \models \psi \\
                (M, w) \models \B{i} \varphi       & \text{ iff } & \text{for all } v \text{ if } w R_i v \text{ then } (M, v) \models \varphi
            \end{array}
        \end{equation*}

    \end{definition}
    We say that two pointed models $\mathcal{M}$ and $\mathcal{M}'$ \emph{agree on} (the formulas of) a set $\Phi \subseteq \Lang$ if, for all $\phi \in \Phi$, $\mathcal{M} \models \phi$ iff $\mathcal{M}' \models \phi$.
    We recall below the notions of bisimulation and bounded bisimulation ($k$-bisimulations) \cite{book/cup/Blackburn2001,books/el/07/GorankoO}.
    
    \begin{definition}
        Let $(M, w_d)$ and $(M', w'_d)$ be two pointed models, with $M=(W, R, V)$ and $M' = (W', R', V')$. A \emph{bisimulation} between $(M,w_d)$ and $(M',w'_d)$ is a non-empty binary relation $Z \subseteq W \times W'$ with $(w_d, w'_d) \in Z$ and satisfying:
        \begin{compactitem}
            \item {[atom]} If $(w, w') \in Z$, then for all $\atom{p} \in \atomSet$, $w \in V(\atom{p})$ iff $w' \in V'(\atom{p})$.
            \item {[forth]} If $(w, w') \in Z$ and $w R_i v$, then there exists $v' \in W'$ such that $w' R'_i v'$ and $(v, v') \in Z$.
            \item {[back]} If $(w, w') \in Z$ and $w' R'_i v'$, then there exists $v \in W$ such that $w R_i v$ and $(v, v') \in Z$.
        \end{compactitem}
        If a bisimulation between $(M,w_d)$ and $(M',w'_d)$ exists, we say that $(M,w_d)$ and $(M',w'_d)$ are \emph{bisimilar}, denoted $(M,w_d) \bisim (M',w_d')$.
        When $(M,w) \bisim (M,w')$ for some worlds $w,w'$ of the same model $M$, we simply write $w \bisim w'$, and say that $w$ and $w'$ are \emph{bisimilar}. 
        Finally, we denote the \emph{bisimulation (equivalence) class} of a world $w \in W$ as $[w]_{\bisim} = \{v \in W \mid w \bisim v\}$.
    \end{definition}

    \begin{proposition}[\cite{book/cup/Blackburn2001}]\label{prop:bisim}
        Two pointed models are bisimilar iff they agree on $\Lang$.
    \end{proposition}

    \begin{definition}\label{def:k-bisim}
        Let $k \geq 0$ and let $(M, w_d)$ and $(M', w'_d)$ be two pointed models, with $M=(W, R, V)$ and $M' = (W', R', V')$. A \emph{$k$-bisimulation} between $(M, w_d)$ and $(M', w'_d)$ is a sequence of non-empty binary relations $Z_k \subseteq \cdots \subseteq Z_0 \subseteq W \times W'$ with $(w_d, w'_d) \in Z_k$ and satisfying, for all $h < k$:
        \begin{compactitem}
            \item {[atom]} If $(w, w') \in Z_0$, then for all $\atom{p} \in \atomSet$, $w \in V(\atom{p})$ iff $w' \in V'(\atom{p})$.
            \item {[forth$_h$]} If $(w, w') \in Z_{h+1}$ and $w R_i v$, then there exists $v' \in W'$ such that $w' R'_i v'$ and $(v, v') \in Z_h$.
            \item {[back$_h$]} If $(w, w') \in Z_{h+1}$ and $w' R'_i v'$, then there exists $v \in W$ such that $w R_i v$ and $(v, v') \in Z_h$.
        \end{compactitem}
        If a $k$-bisimulation between $(M, w_d)$ and $(M', w'_d)$ exists, we say that $(M, w_d)$ and $(M', w'_d)$ are \emph{$k$-bisimilar}, denoted $(M, w_d) \bisim_k (M', w'_d)$.
        When $(M,w) \bisim_k (M,w')$ for some worlds $w,w'$ of the same model $M$, we simply write $w \bisim_k w'$, and say that $w$ and $w'$ are \emph{$k$-bisimilar}. 
        Finally, we denote the \emph{$k$-bisimulation (equivalence) class} of a world $w \in W$ as $\class{w}{k} = \{v \in W \mid w \bisim_k v\}$.
    \end{definition}
    Note that a $k$-bisimulation between pointed models is also an $h$-bisimulation for all $h \leq k$, and hence that $k$-bisimilar worlds are also $h$-bisimilar for all $h \leq k$.


    \begin{proposition}[\cite{book/cup/Blackburn2001}]\label{prop:k-bisim}
        Two pointed models are $k$-bisimilar iff they agree on $\{\phi \in \Lang \mid \md(\phi) \leq k\}$, \ie on all of formulas up to modal depth $k$.
    \end{proposition}







        \begin{definition}\label{def:depth}
            Let $(M,w_d)$ be a pointed model. The \emph{depth} $\depth{w}$ of a world $w$ is the length of the shortest path from $w_d$ to $w$ ($\infty$ if no such path exists).
            Given $k \geq 0$, the \emph{restriction} $M \upharpoonright k$ of $M$ to $k$ is the sub-model containing all worlds with depth at most $k$ (and preserving all edges between them). 
        \end{definition}
    
    
        \begin{lemma}[\cite{book/cup/Blackburn2001}]\label{lem:restriction}
            Let $M$ and $k$ be as above. Then, for every world $w$ of $M \upharpoonright k$, we have $(M \upharpoonright k, w) \bisim_{k - \depth{w}} (M, w)$.
        \end{lemma}

    \section{Defining Rooted $k$-Contractions}\label{sec:world-min-contr}
    The notion of bisimulation contraction is well-known in modal logic. The \emph{(bisimulation) contraction} of a pointed model $\mathcal{M} = ((M, R, V), w_d)$, that we denote with $\contr{\mathcal{M}}{}$, is
    defined as the \emph{quotient structure} of $\mathcal{M}$ with respect to $\bisim$, \ie 
    $\contr{\mathcal{M}}{} = ((W', R', V'), [w_d]_\bisim)$,
    where $W' = \{\class{w}{\bisim} \mid w \in W\}$, $R'_i = \{(\class{w}{\bisim}, \class{v}{\bisim}) \mid w R_i v\}$, and $V'(p) = \{\class{w}{\bisim} \in W' \mid w \in V(p)\}$~\cite{books/el/07/GorankoO}. It is relatively straightforward to prove that:
    \begin{inparaenum}
        \item $\contr{\mathcal{M}}{}$ is bisimilar to $\mathcal{M}$; and
        \item $\contr{\mathcal{M}}{}$ is a minimal model bisimilar to $\mathcal{M}$. 
    \end{inparaenum}
    A similar definition exists for \emph{$k$-(bisimulation) contractions}.
    Namely, the $k$-contraction of $\mathcal{M}$, that we denote with $\contr{\mathcal{M}}{k}$, has been defined as the quotient structure of $\mathcal{M}$ with respect to  $\bisim_k$~\cite{journals/logcom/2023/BolanderL,conf/ijcai/Yu2013}.
    We call this the \emph{standard $k$-contraction} of $\mathcal{M}$.
    However, although such a contracted model is $k$-bisimilar to the original one~\cite{journals/logcom/2023/BolanderL,conf/ijcai/Yu2013},
    in general it is not minimal, as the following example shows.

    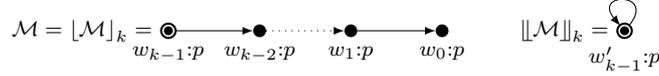
\begin{figure}
        \centering
        \begin{tikzpicture}[>=latex,auto,style={font=\sffamily\footnotesize}]
    \node[poinetdworld, label={below:$w_{k-1}{:}p$}, label={left:$\mathcal{M} = \contr{\mathcal{M}}{k} =$}]
                                                     (w0)  at (0,   0) {};
    \node[world,        label={below:$w_{k-2}{:}p$}] (w1)  at (1.2, 0) {};
    \node[world,        label={below:$w_1{:}p$}]     (w2)  at (2.4, 0) {};
    \node[world,        label={below:$w_0{:}p$}]     (w3)  at (3.6, 0) {};

    \path
        (w0) edge[->]        (w1)
        (w1) edge[->,dotted] (w2)
        (w2) edge[->]        (w3)
    ;
    
    \node[poinetdworld, label={below:$w'_{k-1}{:}p$}, label={left:$\rootedContr{\mathcal{M}}{k}{}=$}] (v0) at (6, 0) {};

    \path
        (v0) edge[->, loop, out=50, in=130, looseness=14] (v0)
    ;
\end{tikzpicture}
        \caption{Standard ($\contr{\mathcal{M}}{k}$) and rooted ($\rootedContr{\mathcal{M}}{k}{}$) $k$-contractions of chain $\mathcal{M}$ (symbol $\rootedContr{\mathcal{M}}{k}{}$ is borrowed from Definition \ref{def:rooted-k-contr-1}).
        Each world $w$ is denoted by a bullet labeled by its name, followed by the atomic propositions that hold in $w$.
        An arrow labeled with $i$ from $w$ to $v$ means that $w R_i v$.
        We omit the labels on arrows whenever $|\modalitiesSet| = 1$.
        The designated world is represented by a circled bullet.
        }
        \label{fig:chain}
    \end{figure}
    \begin{example}\label{ex:chain}
    Consider the chain model $\mathcal{M}$ in Figure \ref{fig:chain} (left). Since $p$ is true in all worlds, and the length of the chain is $k$, a minimal model $k$-bisimilar to $\mathcal{M}$ is a singleton pointed model with a loop (Figure \ref{fig:chain}, right). This is because the loop model preserves all formulas up to depth $k$, cf.\ Proposition~\ref{prop:k-bisim}. However, the standard $k$-contraction of $\mathcal{M}$ is simply $\mathcal{M}$ itself: First note that for all $h \leq k-1$, the formula $\Diamond^{h} \Box \bot$ is true only in world $w_h$ of $\mathcal{M}$ (it expresses the existence of a path of length $h$ to a world from which no world is accessible). Hence any two worlds of $\mathcal{M}$ can be distinguished by a formula $\Diamond^{h} \Box \bot$ of depth $h \leq k-1$. This implies that no two distinct worlds of $\mathcal{M}$ are modally equivalent to modal depth $k$, and hence cannot be part of the same $k$-bisimulation class. Thus, the standard $k$-contraction of $\mathcal{M}$ is $\mathcal{M}$ itself. 
    \end{example}

    We now move to introduce our \emph{rooted $k$-contractions}.
    First, in this section, we show how to define a notion of a rooted $k$-contraction that guarantees the resulting model to have the smallest number of worlds among any model $k$-bisimilar to the original one (we call this property \emph{world minimality}). Later we then define a stronger notion of rooted $k$-contraction that additionally guarantees the set of edges of the contracted model to be minimal (called \emph{edge minimality}).  In what follows, we fix a constant $k \geq 0$ and a pointed model $\mathcal{M} = (M, w_d)$ with $M = (W,V,R)$. Recall the notion of the depth $\depth{w}$ of a world $w$ (Definition \ref{def:depth}). We now introduce the notion of \emph{bound} of a world.

    \begin{definition}\label{def:bound}
        The \emph{bound} of a world $w$ is $\bound{w} = k - \depth{w}$.
    \end{definition}

    \begin{lemma}\label{lem:bound}
        If $x R_i y$, then $\bound{y} \geq \bound{x} - 1$.
    \end{lemma}

    \begin{proof}
        $\depth{y}$ is the length of the shortest path from the designated world to $y$.
        Such a path either goes through the edge $(x,y) \in R_i$, or there is a shorter path to $y$. Hence, $\depth{y} \leq \depth{x} + 1 \Leftrightarrow k - \bound{y} \leq k - \bound{x} + 1 \Leftrightarrow \bound{y} \geq \bound{x} - 1$.
    \end{proof}

    \begin{figure}
        \centering
        \begin{tikzpicture}[>=latex,node distance=1.7em and 1.5em,style={font=\sffamily\footnotesize},yscale=1]
    \begin{scope}[inner sep=0.5pt]
        \node[poinetdworld,                    label={above right:$w_d{:}p$}] (wd) {};
        \node[world,        below left =of wd, label={below left :$w_1{:}p$}] (w1) {};
        \node[world,        below right=of wd, label={below right:$w_2{:}p$}] (w2) {};
        \node[world,        below=of w1,       label={below left :$w_3{:}q$}] (w3) {};
        \node[world,        below=of w2,       label={below right:$w_4{:}q$}] (w4) {};
    \end{scope}

    \path
        (wd) edge[->]  (w1)
        (wd) edge[->]  (w2)
        (w1) edge[<->] (w3)
        (w2) edge[->]  (w4)
        (w3) edge[->]  (w2)
        (w2) edge[->, loop, out=110, in=30, looseness=14] (w2)
    ;

    \begin{scope}[on background layer]
        \path[thick, dotted, color=gray]
            let
                \p{wd} = (wd),
                \p{w1} = ($(w1) - (0.6,0)$),
                \p{w2} = ($(w2) + (0.6,0)$),
                \p{w3} = ($(w3) - (0.6,0)$),
                \p{w4} = ($(w4) + (0.6,0)$)
            in
                (\x{w1}, \y{wd}) edge (\x{w2}, \y{wd})
                (\x{w1}, \y{w1}) edge (\x{w2}, \y{w1})
                (\x{w3}, \y{w3}) edge (\x{w4}, \y{w3})
                ;
    \end{scope}


    \begin{scope}[inner sep=0pt]
        \coordinate (col_1) at (-6.2, 0  );
        \coordinate (col_2) at (-5.2, 0  );
        \coordinate (col_3) at (-4.0, 0  );
        \coordinate (col_4) at (-2.8, 0  );
        \coordinate (row_1) at ( 0  , 0.5);

        \node (d)  at ($(wd) + (col_1) + (row_1)$) {$d$  };
        \node (b1) at ($(wd) + (col_2) + (row_1)$) {$b\ (k {=} 1)$};
        \node (b2) at ($(wd) + (col_3) + (row_1)$) {$b\ (k {=} 2)$};
        \node (b3) at ($(wd) + (col_4) + (row_1)$) {$b\ (k {=} 3)$};

        \node                 (d0)  at ($(wd) + (col_1)$) {$0$ };
        \node[below = of d0]  (d1)                        {$1$ };
        \node[below = of d1]  (d2)                        {$2$ };

        \node                 (b10) at ($(wd) + (col_2)$) {$1$ };
        \node[below = of b10] (b11)                       {$0$ };
        \node[below = of b11] (b12)                       {$-1$};

        \node                 (b20) at ($(wd) + (col_3)$) {$2$ };
        \node[below = of b20] (b21)                       {$1$ };
        \node[below = of b21] (b22)                       {$0$ };

        \node                 (b30) at ($(wd) + (col_4)$) {$3$ };
        \node[below = of b30] (b31)                       {$2$ };
        \node[below = of b31] (b32)                       {$1$ };
        
        \path[thin]
            let
                \p1 = ($(d)  !0.5!(d0)  - ( 0.3, 0  )$),
                \p2 = ($(b3) !0.5!(b30) + ( 0.6, 0  )$),
                \p3 = ($(d)  !0.5!(b1)  + (-0.1, 0.2)$),
                \p4 = ($(d2) !0.5!(b12) - ( 0.1, 0.3)$),
                \p5 = ($(b1) !0.5!(b2)  + ( 0  , 0.2)$),
                \p6 = ($(b12)!0.5!(b22) - ( 0  , 0.3)$),
                \p7 = ($(b2) !0.5!(b3)  + ( 0  , 0.2)$),
                \p8 = ($(b22)!0.5!(b32) - ( 0  , 0.3)$)
            in
                (\p1) edge (\p2)
                (\p3) edge (\p4)
                (\p5) edge (\p6)
                (\p7) edge (\p8)
                ;
    \end{scope}
\end{tikzpicture}
        \caption{Depth ($d$) and bound ($b$) of worlds for $k = 1, 2$ and $3$.
        }
        \label{fig:bound}
    \end{figure}
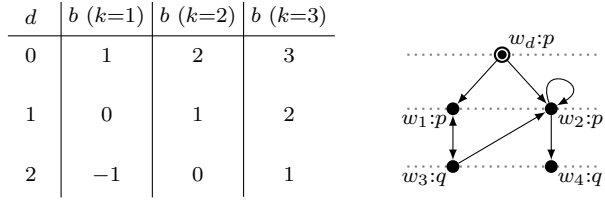

    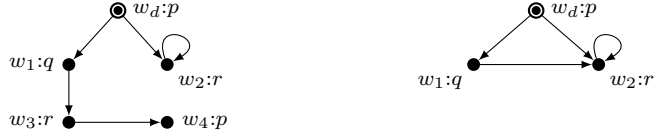
\begin{figure}
        \centering
        \subfloat{
            \begin{tikzpicture}[y = 9.5em,>=latex,node distance=1.7em and 1.5em,style={font=\sffamily\footnotesize}, baseline=(wd)]
    \begin{scope}[inner sep=0.5pt,label distance=0.25em]
        \node[poinetdworld,                    label={      right:$w_d{:}p$}] (wd) {};
        \node[world,        below left =of wd, label={      left :$w_1{:}q$}] (w1) {};
        \node[world,        below right=of wd, label={below right:$w_2{:}r$}] (w2) {};
        \node[world,        below=of w1,       label={      left :$w_3{:}r$}] (w3) {};
        \node[world,       below=of w2,       label={ right:$w_4{:}p$}] (w4) {};
    \end{scope}

    \path
        (wd) edge[->] (w1)
        (wd) edge[->] (w2)
        (w1) edge[->] (w3)
        (w2) edge[->, loop, out=110, in=30, looseness=14] (w2)
        (w3) edge[->] (w4)
    ;
\end{tikzpicture}
            \label{fig:dummy-model}
        }
        \hspace{2cm}
        \subfloat{
            \begin{tikzpicture}[y = 9.5em,>=latex,node distance=1.7em and 2em,style={font=\sffamily\footnotesize}, baseline=(wd)]
    \begin{scope}[inner sep=0.5pt,label distance=0.25em]
        \node[poinetdworld,                    label={      right:$w_d{:}p$}] (wd) {};
        \node[world,        below left =of wd, label={      below left :$w_1{:}q$}] (w1) {};
        \node[world,        below right=of wd, label={below right:$w_2{:}r$}] (w2) {};
    \end{scope}

    \path
        (wd) edge[->] (w1)
        (wd) edge[->] (w2)
        (w1) edge[->] (w2)
        (w2) edge[->, loop, out=110, in=30, looseness=14] (w2)
    ;
\end{tikzpicture}
            \label{fig:dummy-contr}
        }
        \caption{Two $2$-bisimilar pointed models: $\mathcal{N}_1$ (left) and $\mathcal{N}_2$ (right).}
        \label{fig:dummy}
    \end{figure}

    \begin{example}\label{exam:bound} 
        The notion of \emph{bound} of a world will play a key role in the definition of rooted $k$-contractions.
        Figure \ref{fig:bound} shows an example of bound of worlds for $k = 1, 2$ and $3$. Now    
        consider the pointed models $\mathcal{N}_1$ and $\mathcal{N}_2$ of Figure \ref{fig:dummy}.
        Taking the standard $2$-contraction of $\mathcal{N}_1$ would result in the model $\mathcal{N}_1 \upharpoonright 2$, \ie the model $\mathcal{N}_1$ with $w_4$ removed,      
        by a similar argument as in Example~\ref{ex:chain}.
        However,  $\mathcal{N}_1 \upharpoonright 2$  is not world minimal among models $2$-bisimilar to $\mathcal{N}_1$. That is true for $\mathcal{N}_2$, however. Any model $2$-bisimilar to $\mathcal{N}_1$ must have at least three worlds (one for each atomic proposition), which is exactly what $\mathcal{N}_2$ has. The $2$-bisimulation between $\mathcal{N}_1$ and $\mathcal{N}_2$ is defined by:
        $Z_2 = \{(w_d, w'_d)\}$; $Z_1 = Z_2 \cup \{(w_1, w'_1), (w_2, w'_2)\}$; and $Z_0 = Z_1 \cup \{(w_3, w'_2)\}$. 
        Notice that $\mathcal{N}_2$ has been obtained by $\mathcal{N}_1$ by redirecting all
        incoming edges of $w_3$ to $w_2$ and deleting the worlds that are no longer reachable from the designated world. 
        In standard bisimulation contractions, we simply identify worlds that are bisimilar, meaning that we can merge them into one world. In the case of $\mathcal{N}_1$ and $\mathcal{N}_2$, we cannot just trivially merge $w_2$ and $w_3$ into one world, as they don't have the same successor worlds (we might be left with $w_4$ as a successor to $w_2$, which destroys $2$-bisimilarity). The idea of redirecting edges rather than merging worlds forms a crucial part of the intuition behind our rooted $k$-contractions. The point is that $w_2$ can be used as a ``representative'' for $w_3$ when we perform the contraction, and hence we can get rid of $w_3$. The reason that $w_2$ works as a representative for $w_3$ is that $w_3$ 
        is at depth $2$, so has bound $0$. Intuitively this means that $w_3$ can be represented by any world that it is $0$-bisimilar to, as to maintain $k$-bisimilarity at the designated worlds, we only need to require $(k-d)$-bisimilarity of the worlds at depth $d$, \ie worlds of bound $b$ only need to be $b$-bisimilar. These intuitions are made formally precise in the following. 
    \end{example}




    \begin{lemma}\label{lem:edges-redirection}
        Let $k \geq 0$, let $(M, w_d)$ be a pointed model, with $M = (W, R, V)$ and let $x,y \in W \setminus \{w_d\}$ be two distinct worlds such that $\bound{x} \geq \bound{y} \geq 0$ and $x \bisim_{\bound{y}} y$. Let $(M', w_d)$, with $M' = (W', R', V')$, be the pointed model obtained by deleting $y$ from $(M, w_d)$ and redirecting its incoming edges to $x$. More precisely:
            $W' = W \setminus \{y\}$;
            $R'_i = (R_i \cap (W' \times W')) \cup \{(w, x) \mid w R_i y\}$;
            $V'(p) = V(p) \cap W'$, for all $p \in \atomSet$.
        %
        Then $(M, w_d) \bisim_k (M', w_d)$.
    \end{lemma}

    \begin{proof}
        To avoid confusion, we remark that, for all $w$, $\bound{w}$ refers to the bound that $w$ has in $\mathcal{M}$ (and not in $\mathcal{M}'$).
        Similarly, $w \bisim_h w'$ means $(M, w) \bisim_h (M, w')$.
        For all $0 \leq h \leq k$, let $Z_h \subseteq W \times W'$ be the following binary relation:
        \begin{equation*}
            Z_h = \{(w, w') \in W \times W' \mid w \bisim_h w', \bound{w} \geq h \text{ and } \bound{w'} \geq h\}
        \end{equation*}
        We now show that $Z_k, \dots, Z_0$ is a $k$-bisimulation between $\mathcal{M}$ and $\mathcal{M}'$.
        Clearly, $Z_k \subseteq \dots \subseteq Z_0$. Also, trivially, $(w_d,w_d) \in Z_k$ (recalling that $w \bisim_h w'$ means $(M, w) \bisim_h (M, w')$).
         We get [atom] since $(w, w') \in Z_0$ implies $w \bisim_0 w'$ and thus $w \in V(p)$ iff $w' \in V'(p)$ (by [atom] of Definition \ref{def:k-bisim}).
 
        We now show [forth$_h$].
        Let $h < k$, $(w, w') \in Z_{h+1}$ and $w R_i v$.
        We need to find a $v' \in W'$ such that $w' R'_i v'$ and $(v, v') \in Z_h$.
        Since $(w, w') \in Z_{h+1}$, we have $w \bisim_{h+1} w'$, $\bound{w} \geq h+1$ and $\bound{w'} \geq h+1$.
        From $w \bisim_{h+1} w'$, there exists $u \in W$ such that $w' R_i u$ and $v \bisim_h u$.
        Since $w R_i v$, $w' R_i u$, $\bound{w} \geq h+1$ and $\bound{w'} \geq h+1$, by Lemma \ref{lem:bound}, we get $\bound{v} \geq h$ and $\bound{u} \geq h$.
        We have two cases.
        \begin{inparaenum}
            \item If $u \neq y$, then by construction of $R'_i$, we get $w' R'_i u$.
                  Since $v \bisim_h u$, $\bound{v} \geq h$ and $\bound{u} \geq h$, letting $v' = u$ we get $(v, v') \in Z_h$.
            \item If $u = y$, then by construction of $R'_i$, we get $w' R'_i x$.
                  From 
                  $\bound{y} = \bound{u} \geq h$ and $x \bisim_{\bound{y}} y$, we get $x \bisim_h y$.
                  Since $v \bisim_h y \bisim_h x$, $\bound{v} \geq h$ and $\bound{y} \geq h$, letting $v' = x$ we get $(v, v') \in Z_h$.
        \end{inparaenum}
        This concludes [forth$_h$].

        Now for [back$_h$].
        Let $h < k$, $(w, w') \in Z_{h+1}$ and $w' R'_i v'$.
        We need to find $v$ such that $w R_i v$ and $(v, v') \in Z_h$.
        Since $(w, w') \in Z_{h+1}$, we have $w \bisim_{h+1} w'$, $\bound{w} \geq h+1$ and $\bound{w'} \geq h+1$.
        We have two cases.
        \begin{inparaenum}
            \item\label{itm:back-1} If $v' \neq x$, then by construction of $R'_i$, we get $w' R_i v'$.
                  Since $w \bisim_{h+1} w'$, there exists $v \in W$ such that $w R_i v$ and $v \bisim_h v'$.
                  As in [forth$_h$], Lemma \ref{lem:bound} gives $\bound{v} \geq h$ and $\bound{v'} \geq h$.
                  Thus, $(v, v') \in Z_h$.
            \item If $v' = x$, then by construction of $R'_i$, we get $w' R_i x$ or $w' R_i y$.
                  If $w' R_i x$, we can reason as in (\ref{itm:back-1}).
                  If $w' R_i y$, pick $v$ with $w R_i v$ and $v \bisim_h y$.
                  As before, $\bound{v} \geq h$ and $\bound{y} \geq h$.
                  Since $\bound{x} \geq \bound{y} \geq h$ and $x \bisim_{\bound{y}} y$, we get $x \bisim_h y$, and thus $v \bisim_h x$.
                  Hence $(v, x) \in Z_h$, as required.
        \end{inparaenum}
    \end{proof}

    The lemma 
    tells us that if we're only interested in preserving $k$-bisimilarity,  
    a world $y$ can be deleted from a model if there exists a distinct world $x$ such that $\bound{x} \geq \bound{y}$ and $x \bisim_{\bound{y}} y$. This leads us to the following definition.

    \begin{definition}\label{def:representative}\label{def:max-repr}
        Let $x,y$ be two worlds with non-negative bound. We say that $x$ \emph{represents} $y$, denoted by $x \succeq y$, iff $\bound{x} \geq \bound{y}$ and $x \bisim_{\bound{y}} y$. If furthermore $\bound{x} > \bound{y}$, we say that $x$ \emph{strictly represents} $y$, denoted by $x \succ y$. 
        The set of \emph{maximal representatives} of $W$ is the set of worlds $\maxRepr{W} = \{x \in W \mid \bound{x} \geq 0 \textnormal{ and } \neg \exists y \in W (y \succ x)\}$.
        We say that a world $x$ is a \emph{maximal representative} of $y$ if $x \in \maxRepr{W}$ and $x \succeq y$. 
    \end{definition}
    Note that every world $w\in W$ with $b(w) \geq 0$ has at least one maximal representative: Any chain $w \prec w' \prec w'' \prec \cdots$ is finite (since $W$ is finite) and must hence end in a maximal representative of $w$.    
    We are going to build our rooted $k$-contractions on the maximal representatives,  
    the intuition being that all other worlds can be represented by one of these and hence be deleted, cf.\ Lemma~\ref{lem:edges-redirection}. 
    

    \begin{proposition}\label{prop:max-reprs}\label{lem:same-bound-repr}
        For any pointed model $((W, R, V), w_d)$, we have:
        \begin{inparaenum}[1)]
            \item\label{itm:max-reprs-1} $w_d \in \maxRepr{W}$;
            \item\label{itm:max-reprs-2} if $w \succ v$, then $v \notin \maxRepr{W}$; 
            \item\label{itm:max-reprs-3} if $\bound{w} < 0$, then $w \notin \maxRepr{W}$;
            \item\label{itm:max-reprs-4} if $w,v \in \maxRepr{W}$ and $w \bisim_{\bound{w}} v$ then $\bound{w} = \bound{v}$.
        \end{inparaenum}
    \end{proposition}

    \begin{proof}
        Item \ref{itm:max-reprs-1} follows since $w_d$ has bound $\bound{w_d} = k$ and there can not be any world with a greater bound. Items \ref{itm:max-reprs-2} and \ref{itm:max-reprs-3} immediately follow by definition of $\maxRepr{W}$. Item \ref{itm:max-reprs-4} is by contradiction: Suppose that $w \bisim_{\bound{w}} v$ and $\bound{w} > \bound{v}$ (the case $\bound{v} > \bound{w}$ being symmetric).
        Since $w \bisim_{\bound{w}} v$, it follows that $w \bisim_{\bound{v}} v$, which implies that $w \succ v$, contradicting the fact that $v \in \maxRepr{W}$.  
    \end{proof}


    \begin{example}\label{ex:max-reprs}
        Let $\mathcal{M}$ the pointed model in Figure \ref{fig:bound}, let $W$ be its set of worlds and let $k = 2$. Using Proposition~\ref{prop:k-bisim}, we can show that for all worlds $w,v \in W$, if $\bound{w} > \bound{v}$ then $w \not\bisim_{\bound{v}} v$: $w_3$ and $w_4$ of bound 0 are not propositionally equivalent to any world of greater bound (they are the only ones satisfying $p$); $w_1$ and $w_2$ of bound 1 both satisfy $\D{}q$ of modal depth 1, which is not satisfied by the only world of greater bound, $w_d$. Hence, all worlds are maximal representatives, \ie $\maxRepr{W} = W$. We immediately get $\maxRepr{W} = W$ for $k=3$ as well, since if worlds $w$ and $v$ are not $n$-bisimilar, they are also not $(n+1)$-bisimilar.  
%
%
%
        For $k=1$, Proposition \ref{prop:max-reprs}(\ref{itm:max-reprs-3}) gives
 $w_3, w_4 \notin \maxRepr{W}$, and since $w_d \succ w_1,w_2$ (they satisfy the same atomic propositions), we get $\maxRepr{W} = \{w_d\}$.
    \end{example}
    

    \begin{definition}\label{def:repr-class}
        The \emph{representative class} of a world $w$ is the class $\class{w}{\bound{w}}$, which we denote with the compact notation $\reprClass{w}$.
    \end{definition}




    \begin{definition}\label{def:rooted-k-contr-1}
        Let $\mathcal{M} = ((W, R, V), w_d)$ and let $k \geq 0$. The \emph{rooted $k$-contraction} of $\mathcal{M}$ is the pointed model $\rootedContr{\mathcal{M}}{k}{} = ((W', R', V'), \reprClass{w_d})$, where:
        \begin{compactitem}
            \item $W'    = \{\reprClass{x} \mid x \in \maxRepr{W}\}$;
            \item $R'_i  = \{(\reprClass{x}, \reprClass{y}) \mid x,y \in \maxRepr{W}, \exists z (x R_i z$ and $y \bisim_{\bound{x}{-}1} z)$ and $\bound{x} > 0\}$;\footnote{The definition of $R_i'$ doesn't depend on the choice of maximal representatives: 
            If $\reprClass{x} = \reprClass{x'}$ and $\reprClass{y} = \reprClass{y'}$ and $x,x',y,y' \in \maxRepr{W}$, then $b(x) > 0$ iff $b(x') > 0$ and $\exists z(x R_i z \text{ and } y \bisim_{b(x)-1} z)$ iff $\exists z'(x' R_i z' \text{ and } y' \bisim_{b(x')-1} z')$. 
            To prove this, first note that since $\reprClass{x} = \reprClass{x'}$, we get $x \in \reprClass{x'}$ and $x' \in \reprClass{x}$. Similarly for $y$ and $y'$. Since $x \in \reprClass{x'} = [x']_{b(x')}$, we then get $x \bisim_{b(x')} x'$, and hence $b(x) = b(x')$, 
            by Proposition~\ref{prop:max-reprs}(\ref{itm:max-reprs-4}). Hence, $b(x) > 0$ iff $b(x') > 0$. Now suppose $\exists z(x R_i z \text{ and } y \bisim_{b(x)-1} z)$. We need to show that $\exists z'(x' R_i z' \text{ and } y' \bisim_{b(x')-1} z')$. Since $x \bisim_{b(x')} x'$ and $x R_i y$, there exists a $z'$ such that $x' R_i z'$ and $y \bisim_{b(x')-1} z'$. Since 
            $x R_i y$ then $\bound{y} \geq \bound{x} -1$, by Lemma~\ref{lem:bound}. As $\bound{x} = \bound{x'}$, we get $\bound{y} \geq \bound{x'} -1$. Since $y' \in \reprClass{y} = [y]_{b(y)}$, we get $y' \bisim_{\bound{y}} y$ and hence $y' \bisim_{\bound{x'}-1} y$. Combining $y \bisim_{b(x')-1} z'$ and $y' \bisim_{\bound{x'}-1} y$, we get $y' \bisim_{b(x')-1} z'$, as required.}
            \item $V'(p) = \{\reprClass{x} \mid x \in \maxRepr{W} \text{ and } x \in V(p)\}$.\footnote{The definition of $V'(p)$ is well-defined since from $x' \in \reprClass{x}$ we get that $x' \bisim_{\bound{x}} x$ and, thus, $x \in V(p)$ iff $x' \in V'(p)$, by [atom] of Definition \ref{def:k-bisim}.}
        \end{compactitem}
    \end{definition}

    The definition of $R'_i$ requires some explanation.
    At first, one might think that defining the set of $i$-edges as $\{(\reprClass{x}, \reprClass{y}) \mid x,y \in \maxRepr{W}$ and $x R_i y\}$ would be sufficient.
    However, this is not the case. To show this, consider the pointed model $\mathcal{N}_1$ in Figure \ref{fig:dummy} and let $k = 2$. 
    One can easily show that $\maxRepr{W} = \{w_d, w_1, w_2\}$, hence the rooted $2$-contraction has worlds  
    $W' = \{\reprClass{w_d}, \reprClass{w_1}, \reprClass{w_2}\}$.
    Since $(w_1, w_2) \notin R$, defining the accessibility relation of the contracted model by $R' =  \{(\reprClass{x}, \reprClass{y}) \mid x,y \in \maxRepr{W}$ and $x R y\}$ would imply $(\reprClass{w_1}, \reprClass{w_2}) \notin R'$. 
    But then the formula $\D{} r$ would be true in $w_1$ and not in $\reprClass{w_1}$, \ie $w_1 \not\bisim_1 \reprClass{w_1}$, implying that $\mathcal{N}_1$ is not even $2$-bisimilar to its own rooted $2$-contraction! With our current definition of $R'$, we actually get that the rooted $2$-contraction of $\mathcal{N}_1$ is exactly the model $\mathcal{N}_2$ of Figure~\ref{fig:dummy} that we in Example~\ref{exam:bound}  showed to be a world-minimal model $2$-bisimilar to $\mathcal{N}_1$.    


    \begin{lemma}\label{lemm:relation-extension}
        Let $\mathcal{M} = ((W,R,V),w_d)$ and 
        $\rootedContr{\mathcal{M}}{k}{} = ((W', R', V'), w'_d)$. For all $i \in \modalitiesSet$ we have $R_i' \supseteq \{(\reprClass{x},\reprClass{y}) \mid x,y \in \maxRepr{W}, x R_i y \text{ and } b(x) > 0 \}$. 
    \end{lemma}
    \begin{proof}
        Let $x,y \in \maxRepr{W}$ be such that $x R_i y$ and $b(x) >0$. We need to show that $(\reprClass{x}, \reprClass{y}) \in R'_i$. By definition of $R'_i$, this is true if there exists a $z \in W$ such that $x R_i z$ and $y \bisim_{\bound{x}{-}1} z$.  
        Since $x R_i y$, we can take $z = y$. 
    \end{proof}

    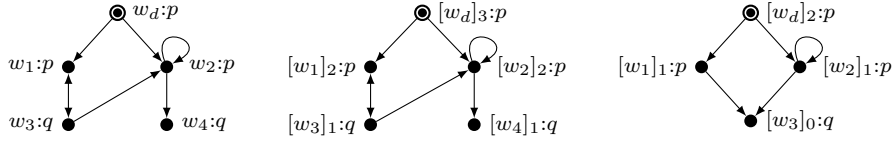
\begin{figure}
        \centering
        \subfloat{
            \begin{tikzpicture}[>=stealth',node distance=1.7em and 1.5em,style={font=\sffamily\footnotesize}, baseline=(wd),label distance=0.25em]
    \begin{scope}[inner sep=0.5pt]
        \node[poinetdworld,                    label={right:$w_d{:}p$}] (wd) {};
        \node[world,        below left =of wd, label={left :$w_1{:}p$}] (w1) {};
        \node[world,        below right=of wd, label={[xshift=1mm]right:$w_2{:}p$}] (w2) {};
        \node[world,        below=of w1,       label={left :$w_3{:}q$}] (w3) {};
        \node[world,        below=of w2,       label={right:$w_4{:}q$}] (w4) {};
    \end{scope}

    \path
        (wd) edge[-latex]  (w1)
        (wd) edge[-latex]  (w2)
        (w1) edge[latex-latex] (w3)
        (w2) edge[-latex]  (w4)
        (w3) edge[-latex]  (w2)
        (w2) edge[-latex, loop, out=110, in=30, looseness=14] (w2)
    ;
 \end{tikzpicture}
            \label{fig:ex-model-1}
        }
        \hfill
        \subfloat{
            \begin{tikzpicture}[y = 9.5em,>=latex,node distance=1.7em and 1.6em,style={font=\sffamily\footnotesize}, baseline=(wd),label distance=0.25em]
    \begin{scope}[inner sep=0.5pt]
        \node[poinetdworld,                    label={      right:$\class{w_d}{3}{:}p$}] (wd) {};
        \node[world,        below left =of wd, label={left :$\class{w_1}{2}{:}p$}] (w1) {};
        \node[world,        below right=of wd, label={[xshift=1mm]right:$\class{w_2}{2}{:}p$}] (w2) {};
        \node[world,        below=of w1,       label={left :$\class{w_3}{1}{:}q$}] (w3) {};
        \node[world,        below=of w2,       label={right:$\class{w_4}{1}{:}q$}] (w4) {};
    \end{scope}

    \path
        (wd) edge[->]  (w1)
        (wd) edge[->]  (w2)
        (w1) edge[<->] (w3)
        (w2) edge[->]  (w4)
        (w3) edge[->]  (w2)
        (w2) edge[->, loop, out=110, in=30, looseness=14] (w2)
    ;
\end{tikzpicture}
            \label{fig:ex-3-contr-1}
        }
        \hfill
        \subfloat{
            \begin{tikzpicture}[>=latex,node distance=1.7em and 1.5em,style={font=\sffamily\footnotesize}, baseline=(wd),label distance=0.25em]
    \begin{scope}[inner sep=0.5pt]
        \node[poinetdworld,                    label={      right:$\class{w_d}{2}{:}p$}] (wd) {};
        \node[world,        below left =of wd, label={left :$\class{w_1}{1}{:}p$}] (w1) {};
        \node[world,        below right=of wd, label={[xshift=1mm]right:$\class{w_2}{1}{:}p$}] (w2) {};
        \node[world,        below right=of w1, label={right      :$\class{w_3}{0}{:}q$}] (w3) {};
    \end{scope}

    \path
        (wd) edge[->] (w1)
        (wd) edge[->] (w2)
        (w1) edge[->] (w3)
        (w2) edge[->] (w3)
        (w2) edge[->, loop, out=110, in=30, looseness=14] (w2)
    ;
\end{tikzpicture}
            \label{fig:ex-2-contr-1}
        }
        \caption{Pointed model $\mathcal{M}$ (left) of Figure~\ref{fig:bound}, $\rootedContr{\mathcal{M}}{3}{}$ (center) and $\rootedContr{\mathcal{M}}{2}{}$ (right).}
        \label{fig:ex-contr-1}
    \end{figure}
    \begin{example}\label{ex:k-contr-1}
        Let $\mathcal{M} = (M, w_d)$ with $M = (W, R, V)$ be the pointed model of Figure \ref{fig:bound}, shown again in Figure~\ref{fig:ex-contr-1} (left). 
        Let $\rootedContr{\mathcal{M}}{3}{} = (M', w'_d)$, with $M' = (W', R', V')$, shown in Figure~\ref{fig:ex-contr-1} (center).
        In Example \ref{ex:max-reprs}, we showed that, when $k=3$, we have $\maxRepr{W} = W$.
        From Definition \ref{def:rooted-k-contr-1}, we then have $W' = \{\class{w_d}{3}, \class{w_1}{2}, \class{w_2}{2}, \class{w_3}{1}, \class{w_4}{1}\}$ and $w'_d = \class{w_d}{3}$, where $\class{w_d}{3} = \{w_d\}$, $\class{w_1}{2} = \{w_1\}$, $\class{w_2}{2} = \{w_2\}$, $\class{w_3}{1} = \{w_3\}$ and $\class{w_4}{1} = \{w_4\}$.  
        Since $\maxRepr{W} = W$ and $\reprClass{w} = w$ for all $w \in W$, it is easy to check that we also get $R' = \{ \reprClass{x}, \reprClass{y} \mid x R y \}$. 
        Let now $\rootedContr{\mathcal{M}}{2}{} = (M', w'_d)$, with $M' = (W', R', V')$, shown in Figure~\ref{fig:ex-contr-1} (right).
        Again, from Example \ref{ex:max-reprs}, we have $\maxRepr{W} = W$.
        From Definition \ref{def:rooted-k-contr-1}, we then have $W' = \{\class{w_d}{2}, \class{w_1}{1}, \class{w_2}{1}, \class{w_3}{0}\}$ and $w'_d = \class{w_d}{2}$, where $\class{w_d}{2} = \{w_d\}$, $\class{w_1}{1} = \{w_1\}$, $\class{w_2}{1} = \{w_2\}$ and $\class{w_3}{0} = \{w_3, w_4\}$. In this case we get $R' = \{(\reprClass{x}, \reprClass{y}) \mid x R y$ and $\bound{x} > 0\}$. 
    \end{example}

    \section{Properties of Rooted $k$-Contractions}

    The first crucial property to show is that rooted $k$-contractions are $k$-bisimilar to their original models.

    \begin{theorem}\label{th:k-bisim-1}
        Let $\mathcal{M}$ be a pointed model and let $k \geq 0$. Then, $\mathcal{M} \bisim_k \rootedContr{\mathcal{M}}{k}{}$.
    \end{theorem}
        \begin{proof}
            Let $\mathcal{M} = ((W, R, V), w_d)$ and $\rootedContr{\mathcal{M}}{k}{} = ((W', R', V'), w'_d)$. For all $0 \leq h \leq k$, let $Z_h \subseteq W \times W'$ be the following binary relation:
            \begin{equation*}
                Z_h = \{(x, \reprClass{x'}) \mid x' \in \maxRepr{W}, x' \bisim_h x \textnormal{ and } \bound{x} \geq h\}.
            \end{equation*}
            We now show that the sequence $Z_k, \dots, Z_0$ is a $k$-bisimulation between $\mathcal{M}$ and $\rootedContr{\mathcal{M}}{k}{}$.
            Clearly, $Z_k \subseteq \dots \subseteq Z_0$.
            Moreover, since $w_d \in \maxRepr{W}$ (by Proposition \ref{prop:max-reprs}(\ref{itm:max-reprs-1})) and $\bound{w_d} = k$, it follows that $(w_d, \reprClass{w_d}) \in Z_k$.

            We first show [atom].
            Let $(x,x'') \in Z_0$.
            Then, by definition of $Z_0$, we have that $x'' = \reprClass{x'}$ for some $x' \in \maxRepr{W}$ such that $x' \bisim_0 x$ and $\bound{x} \geq 0$.
            From $x'' = \reprClass{x'}$, we get $x'' \bisim_{b(x')} x'$ and hence $x'' \bisim_0 x'$.
            Thus, $x'' \bisim_0 x' \bisim_0 x$ and, by [atom] of Definition \ref{def:k-bisim}, we then get $x'' \in V'(p)$ iff $x \in V(p)$, as required.

            Before moving to [forth$_h$] and [back$_h$], we show the following claim:
            \emph{Let $h \leq k$, $x \in W$ and $x' \in \maxRepr{W}$ be such that $x \bisim_h x'$ and $\bound{x} \geq h$. Then $\bound{x'} \geq h$} ($\dagger$).
            Proof: Since $x' \in \maxRepr{W}$, we have $x \not\succ x'$, hence either $\bound{x} \not> \bound{x'}$ or $x \not\bisim_{\bound{x'}} x'$.
            If $\bound{x} \not> \bound{x'}$, we get $\bound{x'} \geq \bound{x} \geq h$, as required.
            If $x \not\bisim_{\bound{x'}} x'$, then since $x \bisim_h x'$, we must have $\bound{x'} > h$.
            This proves ($\dagger$).

            We now show [forth$_h$].
            Let $h < k$, $(x, x'') \in Z_{h+1}$, and $x R_i y$.
            Then $x'' = \reprClass{x'}$, where $x' \in \maxRepr{W}$, $x' \bisim_{h+1} x$ and $\bound{x} \geq h+1$.
            From ($\dagger$) we get $b(x') \geq h+1$.
            We need to find a world $y'' \in W'$ such that $x'' R'_i y''$ and $(y, y'') \in Z_h$.
            Since $x R_i y$ and $x \bisim_{h+1} x'$, there exists $z \in W$ such that $x' R_i z$ and $y \bisim_h z$.
            Let $y'$ be a maximal representative of $z$.
            Then, $y' \in \maxRepr{W}$, $\bound{y'} \geq \bound{z}$ and $y' \bisim_{\bound{z}} z$.
            Since $x' R_i z$, by Lemma \ref{lem:bound} we have $\bound{z} \geq \bound{x'} - 1$.
            Since $y' \bisim_{\bound{z}} z$, then we also get $y' \bisim_{\bound{x'}-1} z$.
            We now have $x', y' \in \maxRepr{W}$, $x' R_i z$, $y' \bisim_{\bound{x'}-1} z$ and $\bound{x'} \geq h+1 > 0$, which by Definition \ref{def:rooted-k-contr-1} means $\reprClass{x'} R'_i \reprClass{y'}$.
            Letting $y'' = \reprClass{y'}$, we have hence found a $y''$ such that $x'' R'_i y''$.
            The only thing left to prove now is that $(y, y'') \in Z_h$. By definition of $Z_h$, it suffices to prove $y' \in \maxRepr{W}$, $y' \bisim_h y$ and $\bound{y} \geq h$.
            We already have $y' \in \maxRepr{W}$.
            Since $y' \bisim_{\bound{x'}-1} z$ and $\bound{x'} \geq h+1$, we get $y' \bisim_h z$.
            As we also have $y \bisim_h z$, we get $y' \bisim_h y$, as required.
            Finally, since $x R_i y$ and $\bound{x} \geq h+1$, Lemma \ref{lem:bound} gives $\bound{y} \geq h$.
            This concludes [forth$_h$].

            We now show [back$_h$].
            Let $h < k$, $(x, x'') \in Z_{h+1}$, and $x'' R'_i y''$.
            We need to find a world $y \in W$ such that $x R_i y$ and $(y, y'') \in Z_h$.
            Since $x'' R'_i y''$, by Definition \ref{def:rooted-k-contr-1} there exist $x', y' \in \maxRepr{W}$ and $z \in W$ such that $x'' = \reprClass{x'}$, $y'' = \reprClass{y'}$, $x' R_i z$, $y' \bisim_{\bound{x'}-1} z$ and $\bound{x'} > 0$.
            Since $(x, x'') \in Z_{h+1}$, then $x'' = \reprClass{\hat{x}}$, where $\hat{x} \in \maxRepr{W}$, $\hat{x} \bisim_{h+1} x$ and $\bound{x} \geq h+1$.
            By ($\dagger$), we get $\bound{\hat{x}} \geq h+1$.
            Since $x', \hat{x} \in \maxRepr{W}$ and $\reprClass{x'} = \reprClass{\hat{x}}$, by Proposition~\ref{lem:same-bound-repr}(\ref{itm:max-reprs-4}) we get $\bound{x'} = \bound{\hat{x}}$.
            From $x' \bisim_{\bound{x'}} \hat{x}$ and $\bound{x'} \geq h+1$ we get $x' \bisim_{h+1} \hat{x}$ and, hence, $x' \bisim_{h+1} x$.
            Since $x' \bisim_{h+1} x$ and $x' R_i z$, there exists $y \in W$ such that $x R_i y$ and $y \bisim_h z$.
            Only left to show is that $(y, y'') \in Z_h$.
            From $y' \bisim_{\bound{x'}-1} z$ and $\bound{x'} \geq h+1$, we get $y' \bisim_h z$ and, thus, $y' \bisim_h y$.
            Since $x R_i y$ and $\bound{x} \geq h+1$, Lemma \ref{lem:bound} gives $\bound{y} \geq h$.
            We now have $y' \in \maxRepr{W}$, $y' \bisim_h y$ and $\bound{y} \geq h$.
            Thus, $(y, y'') \in Z_h$, as required.
            This concludes [back$_h$].
    \end{proof}
   
    We now prove world minimality.
    To show this property, it is useful to group the worlds of a rooted $k$-contraction wrt.\ to their bound and analyze them separately. Specifically, we prove that each such group of worlds is minimal.
    To this end, we first show an intermediate result, namely that a maximal representative $x$ of $\mathcal{M}$ and its representative class $\reprClass{x}$ have the same bound (wrt.\ $\mathcal{M}$ and $\rootedContr{\mathcal{M}}{k}{}$, respectively).
    This result highlights the link between the notions of maximal representatives and representative classes, since a representative class $\reprClass{x}$ maintains the same bound as the maximal representative $x$.

    \begin{lemma}\label{lem:max-repr-bound}
        Let $\mathcal{M} = ((W, R, V), w_d)$ be a pointed model with rooted $k$-contraction $\rootedContr{\mathcal{M}}{k}{} = ((W', R', V'), w'_d)$ and let $x \in \maxRepr{W}$. Then $\bound{x} = \bound{\reprClass{x}}$.
    \end{lemma}
        \begin{proof}
            The proof is by induction on $h = \bound{x}$.
            For the base case, we consider $h = k$ (we do induction from $h=k$ down to $h=0$).
            By definition, only the designated world of a model has bound $k$, so we immediately get $\bound{w_d} = \bound{\reprClass{w_d}} = k$, concluding the base case.
        Assume now by induction hypothesis (I.H.) that for all $x \in \maxRepr{W}$ with $\bound{x} = h > 0$ we have $\bound{x} = \bound{\reprClass{x}}$.
        Let $y \in \maxRepr{W}$ with $\bound{y} = h-1$.
        We need to show $\bound{\reprClass{y}} = h-1$.
        Since $b(y) = h-1$, there must exist an $x$ with $x R_i y$ and $b(x) = h$. 
        We now prove $x \in \maxRepr{W}$ by contradiction: Assuming $x \notin \maxRepr{W}$,
        there exists $x' \in \maxRepr{W}$ such that $x' \succ x$, \ie $\bound{x'} > \bound{x}$ and $x' \bisim_{\bound{x}} x$. 
        Since $x R_i y$, there exists $y'$ such that $x' R_i y'$ and $y \bisim_{\bound{x}-1} y'$.
        As $\bound{x} = h$ and $\bound{y} = h-1$ we get $y \bisim_{\bound{y}} y'$.
        Since $x' R_i y'$ and $\bound{x'} > \bound{x} = h$, Lemma \ref{lem:bound} gives $\bound{y'} > h-1$ and thus $\bound{y'} > \bound{y}$.
        We now have $\bound{y'} > \bound{y}$ and $y' \bisim_{\bound{y}} y$, which means $y' \succ y$, contradicting $y \in \maxRepr{W}$.
        Thus, $x \in \maxRepr{W}$.
            %
%
            Since $x,y  \in \maxRepr{W}$, $x R_i y$ and $\bound{x} = h > 0$, Lemma~\ref{lemm:relation-extension} gives $\reprClass{x} R'_i \reprClass{y}$.
            Since $\bound{\reprClass{x}} = h$ (by I.H.), Lemma~\ref{lem:bound} then gives $\bound{\reprClass{y}} \geq h-1$. We also have $\bound{\reprClass{y}} \leq h-1$, since if $\bound{\reprClass{y}} \geq h$, then I.H.\ would give $\bound{y} \geq h$, contradicting $\bound{y} = h-1$. Thus $\bound{\reprClass{y}} = h-1$, as required.
%
    \end{proof}
   

    \begin{corollary}\label{cor:non-h-bisim}
        If $x \neq y$ are worlds of $\rootedContr{\mathcal{M}}{k}{}$ with   
        $\bound{x} = \bound{y} = h$ then $x \not\bisim_h y$.
    \end{corollary}

    \begin{proof}
        Let $\rootedContr{\mathcal{M}}{k}{} = (M, w_d)$ and $\mathcal{M} = (M', w'_d)$.
        By Definition \ref{def:rooted-k-contr-1}, we have $x = \reprClass{x'}$ and $y = \reprClass{y'}$ for some worlds $x',y' \in \maxRepr{W}$.
        By Lemma \ref{lem:max-repr-bound} we get $\bound{x} = \bound{x'}$ and $\bound{y} = \bound{y'}$, and hence $\bound{x'} = \bound{y'} = h$.
        Since $x \neq y$ and $\bound{x'} = \bound{y'} = h$, we get $x' \not\bisim_h y'$.
        From the proof of Theorem \ref{th:k-bisim-1} we get $(M', x') \bisim_h (M, x)$ and $(M', y') \bisim_h (M, y)$, and hence $x \not\bisim_h y$.
    \end{proof}
   

    \begin{theorem}\label{th:world-minimal-k-contr}
        Let $\mathcal{M}$ be a pointed model and let $k \geq 0$. Then $\rootedContr{\mathcal{M}}{k}{}$ is a world minimal model $k$-bisimilar to $\mathcal{M}$, \ie it has the least number of worlds among all models $k$-bisimilar to $\mathcal{M}$.
    \end{theorem}

    \begin{proof}
        For any model with world set $W$, let $W_h$ denote the subset of worlds with bound $h$, \ie $W_h = \{w \in W \mid \bound{w} = h\}$.
        Let $\mathcal{M} = ((W, R, V), w_d)$,
         $\mathcal{M}' = \rootedContr{\mathcal{M}}{k}{} = ((W', R', V'), w'_d)$ and 
        let $\mathcal{M}'' = ((W'', R'', V''), w''_d)$ be a world-minimal pointed model with $\mathcal{M}' \bisim_k \mathcal{M}''$.
        We want to show that $|W'| \leq  |W''|$.
        We show this by proving that $| W'_h | \leq | W''_h |$ for all $0 \leq h \leq k$ (noting that $W'_h = \varnothing$ for $h < 0$, since no world of negative bound is in $\maxRepr{W}$, by Proposition~\ref{prop:max-reprs}(\ref{itm:max-reprs-3})). To achieve a contradiction, assume that $| W'_h | > | W''_h |$ for some $h \leq k$. By Corollary \ref{cor:non-h-bisim}, we have $x \not\bisim_h y$ for all distinct $x,y \in W'_h$. Hence, any such distinct $x,y \in W'_h$ can not be $h$-bisimilar to the same world of $W''_h$. Since $| W'_h | > | W''_h |$, there must then exist at least one world $x' \in W'_h$ that is not $h$-bisimilar to any world in $W''_h$. 
        Since $\bound{x'} = h$, there exists a path from $w'_d$ to $x'$ of length $k{-}h$.
        Since $(M', w'_d) \bisim_k (M'', w_d'')$, by repeated applications of [forth], we get a path of length $k{-}h$ from $w''_d$ to a world $x''$ with $(M', x') \bisim_h (M'', x'')$. Since $x''$ is reachable by a path of length $k{-}h$ from $w''_d$, we have $\bound{x''} \geq k {-} (k{-}h) = h$.   
        Since $x'$ is not $h$-bisimilar to any world in $W''_h$, we furthermore get $\bound{x''} > h$.
        This implies that $x''$ is reachable by a path of length $< k{-}h$ from $w''_d$, and by repeated applications of [back], we get a path of length $< k{-}h$ from $w'_d$ to a world $y'$ such that $(M', y') \bisim_h (M'', x'')$.
        We now have $(M', x') \bisim_h (M'', x'') \bisim_h (M', y')$, and since the path has length $< k{-}h$, also $\bound{y'} > h$.  
        By Definition \ref{def:rooted-k-contr-1}, there exist $x,y \in \maxRepr{W}$ such that $x' = \reprClass{x}$ and $y' = \reprClass{y}$ and, by Lemma~\ref{lem:max-repr-bound},   
         we get $\bound{x} = \bound{x'} = h$ and $\bound{y} = \bound{y'} > h$.
        From the proof of Theorem \ref{th:k-bisim-1}, we get $(M, x) \bisim_h (M', x')$ and $(M, y) \bisim_{h} (M', y')$, and hence $(M, x) \bisim_h (M', x') \bisim_h (M', y') \bisim_h (M, y)$. Taken together, this implies $y \succ x$, contradicting $x \in \maxRepr{W}$, and we're done. 
    \end{proof}

    \section{Edge Minimal Contractions}\label{sec:edge-min-contr}
    We have defined rooted $k$-contractions and shown them to be world minimal. 
    However, Definition \ref{def:rooted-k-contr-1} does not guarantee that the resulting $k$-contraction is also edge minimal, as we now exemplify. 
    \begin{example}\label{ex:edge-minimal-intro}
    Let $\mathcal{M} = ((W,R,V),w_d)$ be the pointed model in Figure~\ref{fig:ex-contr-1} left and let $\rootedContr{\mathcal{M}}{3}{} = ((W',R',V'),w'_d)$ be its rooted $3$-contraction (Figure~\ref{fig:ex-contr-1} center). Recall from Example~\ref{ex:k-contr-1} that $\maxRepr{W} = W$ and $\bound{w_3} = 1$.  
    Since $w_3 R w_1$ and $w_3 R w_2$, Lemma~\ref{lemm:relation-extension} hence gives us $\reprClass{w_3}R'\reprClass{w_1}$ and $\reprClass{w_3}R' \reprClass{w_2}$. 
    However, including only one of those edges in $R'$ is sufficient to guarantee $3$-bisimilarity to $\mathcal{M}$: $\bound{w_3} = 1$ and thus $\reprClass{w_3}$ only needs to preserve $1$-bisimilarity to $w_3$.  
    \end{example}
    When not all edges are required, we need to decide which to preserve. 
    To this end, we introduce the notion of \emph{least $h$-representative} of a world.
    \begin{definition}\label{def:least-repr}
        Let $<$ be a total order on $W$ and let $0 \leq h \leq k$. 
        The \emph{least $h$-representative} of $w\in W$ is the world 
        $\leastRepr{w}{h} = \min_< \{ v \in \maxRepr{W} \mid v \bisim_h w \}$. 
    \end{definition}

    The least $h$-representative of a world $w$ is the minimal (wrt.\ $<$) maximal representative $v$ such that $v \bisim_h w$. 
    We now present the revised definition of rooted $k$-contraction guaranteeing minimality among $k$-bisimilar models both in number of worlds and edges (and hence minimality in terms of overall size).  

    \begin{definition}\label{def:rooted-k-contr}
        Let  $\mathcal{M} = ((W, R, V), w_d)$, let $k \geq 0$ and let $<$ be a total order on $W$. The \emph{$k$-contraction} of $\mathcal{M}$ wrt.\ $<$ is the pointed model $\rootedContr{\mathcal{M}}{k}{<} = ((W', R', V'), \reprClass{w_d})$, where:
        \begin{compactitem}
            \item $W'    = \{\reprClass{x} \mid x \in \maxRepr{W}\}$;
            \item $R'_i  = \{(\reprClass{x}, \reprClass{\leastRepr{y}{\bound{x}{-}1}}) \mid x \in \maxRepr{W}, x R_i y \text{ and } \bound{x} > 0\}$;
            \item $V'(p) = \{\reprClass{x} \mid x \in \maxRepr{W} \text{ and } x \in V(p)\}$.
        \end{compactitem}
    \end{definition}
    Well-definedness of the definition (independence of choice of representatives) is guaranteed in the same way as for  Definition \ref{def:rooted-k-contr-1}.
    \begin{figure}
        \centering
        \subfloat{
            \begin{tikzpicture}[>=stealth',node distance=1.7em and 1.5em,style={font=\sffamily\footnotesize}, baseline=(wd),label distance=0.25em]
    \begin{scope}[inner sep=0.5pt]
        \node[poinetdworld,                    label={right:$w_d{:}p$}] (wd) {};
        \node[world,        below left =of wd, label={left :$w_1{:}p$}] (w1) {};
        \node[world,        below right=of wd, label={[xshift=1mm]right:$w_2{:}p$}] (w2) {};
        \node[world,        below=of w1,       label={left :$w_3{:}q$}] (w3) {};
        \node[world,        below=of w2,       label={right:$w_4{:}q$}] (w4) {};
    \end{scope}

    \path
        (wd) edge[-latex]  (w1)
        (wd) edge[-latex]  (w2)
        (w1) edge[latex-latex] (w3)
        (w2) edge[-latex]  (w4)
        (w3) edge[-latex]  (w2)
        (w2) edge[-latex, loop, out=110, in=30, looseness=14] (w2)
    ;
 \end{tikzpicture}
            \label{fig:ex-model}
        }
        \hfill
        \subfloat{
            \begin{tikzpicture}[y = 9.5em,>=latex,node distance=1.7em and 1.5em,style={font=\sffamily\footnotesize}, baseline=(wd),label distance=0.25em]
    \begin{scope}[inner sep=0.5pt]
        \node[poinetdworld,                    label={      right:$\class{w_d}{3}{:}p$}] (wd) {};
        \node[world,        below left =of wd, label={left :$\class{w_1}{2}{:}p$}] (w1) {};
        \node[world,        below right=of wd, label={[xshift=1mm]right:$\class{w_2}{2}{:}p$}] (w2) {};
        \node[world,        below=of w1,       label={left :$\class{w_3}{1}{:}q$}] (w3) {};
        \node[world,        below=of w2,       label={right:$\class{w_4}{1}{:}q$}] (w4) {};
    \end{scope}

    \path
        (wd) edge[->]  (w1)
        (wd) edge[->]  (w2)
        (w1) edge[<->] (w3)
        (w2) edge[->]  (w4)
        (w2) edge[->, loop, out=110, in=30, looseness=14] (w2)
    ;
\end{tikzpicture}
            \label{fig:ex-3-contr}
        }
        \hfill
        \subfloat{
            \begin{tikzpicture}[>=latex,node distance=1.7em and 1.5em,style={font=\sffamily\footnotesize}, baseline=(wd),label distance=0.25em]
    \begin{scope}[inner sep=0.5pt]
        \node[poinetdworld,                    label={      right:$\class{w_d}{2}{:}p$}] (wd) {};
        \node[world,        below left =of wd, label={left :$\class{w_1}{1}{:}p$}] (w1) {};
        \node[world,        below right=of wd, label={[xshift=1mm] right:$\class{w_2}{1}{:}p$}] (w2) {};
        \node[world,        below right=of w1, label={right      :$\class{w_3}{0}{:}q$}] (w3) {};
    \end{scope}

    \path
        (wd) edge[->]  (w1)
        (wd) edge[<->] (w2)
        (w1) edge[->]  (w3)
        (w2) edge[->]  (w3)
    ;
\end{tikzpicture}
            \label{fig:ex-2-contr}
        }
        \caption{Pointed model $\mathcal{M}$ (left) of Example \ref{ex:k-contr}, $\rootedContr{\mathcal{M}}{3}{<}$ (center) and $\rootedContr{\mathcal{M}}{2}{<}$ (right).}
        \label{fig:ex-contr}
    \end{figure}
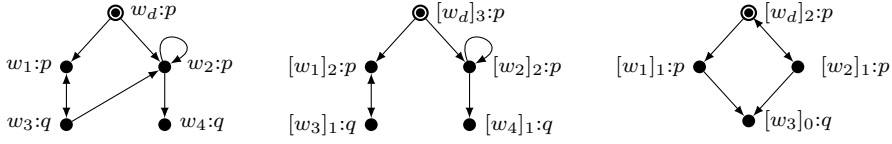
    \begin{example}\label{ex:k-contr}
        Let $\mathcal{M} = ((W,R,V), w_d)$ be the pointed model of Figure \ref{fig:bound} (also shown in Figure~\ref{fig:ex-contr} left). Let $<$ be a total order on $W$ such that $w_d < w_1 < w_2 < w_3 < w_4$. The pointed models $\rootedContr{\mathcal{M}}{3}{<}$ and $\rootedContr{\mathcal{M}}{2}{<}$, also shown in Figure \ref{fig:ex-contr}, will now be analyzed. First 
        let $\rootedContr{\mathcal{M}}{3}{<} = ((W',R',V'), w'_d)$.
        From Example \ref{ex:k-contr-1}, we have $\maxRepr{W} = W$, $W' = \{\class{w_d}{3}, \class{w_1}{2}, \class{w_2}{2}, \class{w_3}{1}, \class{w_4}{1}\}$ and $w'_d = \class{w_d}{3}$.
        Now for the edges of $\rootedContr{\mathcal{M}}{3}{<}$.
        Notice that for all $x R y$ of $\mathcal{M}$ such that $\bound{x} \geq 2$, we have that $\leastRepr{y}{\bound{x}{-}1} = y$.
        Thus, we have $\class{w_d}{3} R' \class{w_1}{2}$, $\class{w_d}{3} R' \class{w_2}{2}$, $\class{w_1}{2} R' \class{w_3}{1}$, $\class{w_2}{2} R' \class{w_2}{2}$ and $\class{w_2}{2} R' \class{w_4}{1}$.
        Finally, consider the edges $w_3 R w_1$ and $w_3 R w_2$
        as discussed in Example~\ref{ex:edge-minimal-intro}. Since $\leastRepr{w_1}{\bound{w_3}{-}1} = \leastRepr{w_2}{\bound{w_3}{-}1} = w_1$, we have $\class{w_3}{1} R' \class{w_1}{2}$ (and not $\class{w_3}{1} R' \class{w_2}{2}$).

        Now let $\rootedContr{\mathcal{M}}{2}{<} = ((W',V',R'), w'_d)$. From Example \ref{ex:k-contr-1}, we have $\maxRepr{W} = W$, $W' = \{\class{w_d}{2}, \class{w_1}{1}, \class{w_2}{1}, \class{w_3}{0}\}$ and $w'_d = \class{w_d}{2}$.
        From $\leastRepr{w_1}{\bound{w_d}{-}1} = w_1$ and $\leastRepr{w_2}{\bound{w_d}{-}1} = w_2$ we get $\class{w_d}{2} R' \class{w_1}{1}$ and $\class{w_d}{2} R' \class{w_2}{1}$.
        From $\leastRepr{w_3}{\bound{w_1}{-}1} = \leastRepr{w_4}{\bound{w_2}{-}1} = w_3$ we get $\class{w_1}{1} R' \class{w_3}{0}$ and $\class{w_2}{1} R' \class{w_3}{0}$.
        Finally, since $\leastRepr{w_2}{\bound{w_2}{-}1} = w_d$, we have $\class{w_2}{1} R' \class{w_d}{2}$.
    \end{example}

    \begin{theorem}\label{th:k-bisim}
        Let $\mathcal{M}$ be a pointed model and let $k \geq 0$. Then $\mathcal{M} \bisim_k \rootedContr{\mathcal{M}}{k}{<}$.
    \end{theorem}
    \begin{proof}
        Let $\mathcal{M} = ((W, R, V), w_d)$ and $\rootedContr{\mathcal{M}}{k}{<} = ((W', R', V'), w'_d)$.
        For all $0 \leq h \leq k$, let $Z_h \subseteq W \times W'$ be as in the proof of Theorem \ref{th:k-bisim-1}:
        \begin{equation*}
            Z_h = \{(x, \reprClass{x'}) \mid x' \in \maxRepr{W}, x' \bisim_h x \textnormal{ and } \bound{x} \geq h\}.
        \end{equation*}
        We now show that $Z_k, \dots, Z_0$ is a $k$-bisimulation between $\mathcal{M}$ and $\rootedContr{\mathcal{M}}{k}{<}$.
        From the proof of Theorem \ref{th:k-bisim-1}, we immediately get $Z_k \subseteq \dots \subseteq Z_0$, $(w_d, \reprClass{w_d}) \in Z_k$ and [atom].
        Moreover, we get [back$_h$].
        To prove this, it suffices to show that $R'_i \subseteq R''_i$, where $((W'', R'', V''), w''_d)$ is the $k$-contraction of $\mathcal{M}$ of Definition \ref{def:rooted-k-contr-1}.
        Let $x'' R'_i y''$.
        Since $x'' R'_i y''$, by Definition \ref{def:rooted-k-contr} there exist $x', y' \in \maxRepr{W}$ and $y \in W$ such that $x'' = \reprClass{x'}$, $y'' = \reprClass{y'}$, $y' = \leastRepr{y}{\bound{x'}{-}1}$, $x' R_i y$ and $\bound{x'} > 0$. From $y' = \leastRepr{y}{\bound{x'}{-}1}$, we get $y' \bisim_{\bound{x'}{-}1} y$. Choosing $z = y$, we then get $x'' R''_i y''$ by Definition \ref{def:rooted-k-contr-1}, 
        as required.

        To show [forth$_h$],
        let $h < k$, $(x, x'') \in Z_{h+1}$, and $x R_i y$.
        Then $x'' = \reprClass{x'}$, where $x' \in \maxRepr{W}$, $x' \bisim_{h+1} x$ and $\bound{x} \geq h+1$.
        We need to find $y'' \in W'$ such that $x'' R'_i y''$ and $(y, y'') \in Z_h$.
        Since $x R_i y$ and $x \bisim_{h+1} x'$, there exists $z \in W$ such that $x' R_i z$ and $y \bisim_h z$.
        Let $y' = \leastRepr{z}{\bound{x'}{-}1}$.
        Then $y' \in \maxRepr{W}$ and $y' \bisim_{\bound{x'}{-}1} z$.
        By ($\dagger$) (see proof of Theorem \ref{th:k-bisim-1}), we get $\bound{x'} \geq h+1$, and thus $y' \bisim_h z \bisim_h y$.
        Since $x R_i y$ and $\bound{x} \geq h+1$, Lemma \ref{lem:bound} gives $\bound{y} \geq h$.
        Let $y'' = \reprClass{y'}$.
        We now have $y' \in \maxRepr{W}$, $y' \bisim_h y$ and $\bound{y} \geq h$, which by definition of $Z_h$ means that $(y, y'') \in Z_h$.
        By Definition \ref{def:rooted-k-contr}, from $x' \in \maxRepr{W}$, $x' R_i z$, $y' = \leastRepr{z}{\bound{x'}{-}1}$ and $\bound{x'} \geq h+1 > 0$, we get $x'' R'_i y''$, as required.
    \end{proof}


    \begin{theorem}\label{th:minimal-k-contr}
        Let $\mathcal{M}$ be a pointed model and $k \geq 0$. Then $\rootedContr{\mathcal{M}}{k}{<}$ is a minimal pointed model $k$-bisimilar to $\mathcal{M}$ (i.e., it is both world and edge minimal). 
    \end{theorem}

    \begin{proof}
        For any model $M = (W,R,V)$, let $W_h$ denote the subset of worlds with bound $h$, \ie $W_h = \{w \in W \mid \bound{w} = h\}$, and let $(R_i)_h$ denote the set of $i$-edges outgoing from worlds in $W_h$, \ie $(R_i)_h = R_i \cap (W_h \times W)$.
        Let $\mathcal{M} = ((W, R, V), w_d)$, 
        $\mathcal{M}' = \rootedContr{\mathcal{M}}{k}{<} = ((W', R', V'), w'_d)$ and
        $\mathcal{M}'' = ((W'', R'', V''), w''_d)$ be a minimal pointed model with $\mathcal{M}' \bisim_k \mathcal{M}''$.
        By Theorem \ref{th:world-minimal-k-contr}, we have $|W''| = |W'|$.
        We want to show that $|R''_i| \leq |R'_i|$ for all $i \in \modalitiesSet$.
        We show this by proving that $|(R'_i)_h| \leq |(R''_i)_h|$ for all $0 < h \leq k$ (noting that $(R'_i)_h = \varnothing$ for all $h \leq 0$, by Definition \ref{def:rooted-k-contr}).
        To achieve a contradiction, assume that $|(R'_i)_h| > |(R''_i)_h|$ for some $i \in \modalitiesSet$ and $0 < h \leq k$.
        Let $x'' \in W''_h$ and let $x'' R''_i y''$.
        Notice that $(x'', y'') \in (R''_i)_h$.
        From the proof of Theorem \ref{th:world-minimal-k-contr}, since $(M', w'_d) \bisim_k (M'', w''_d)$, there exists $x' \in W'_h$ such that $(M', x') \bisim_h (M'', x'')$.
        Since $(M', x') \bisim_h (M'', x'')$, there exist $y'_1, y'_2 \in W'$ such that $(x', y'_1), (x', y'_2) \in (R'_i)_h$, $(M', y'_1) \bisim_{h-1} (M'', y'')$ and $(M', y'_2) \bisim_{h-1} (M'', y'')$.
        Moreover, let $y'_1 \neq y'_2$.
        Notice that, since $|(R'_i)_h| > |(R''_i)_h|$, such $x'', x', y'_1$ and $y'_2$ must exist.
        By Definition \ref{def:rooted-k-contr}, $x' = \reprClass{x}$, $y'_1 = \reprClass{y_1}$ and $y'_2 = \reprClass{y_2}$, for some $x, y_1, y_2 \in \maxRepr{W}$.
        By Lemma \ref{lem:max-repr-bound}, we get $\bound{x} = \bound{x'} = h$ (since $x' \in W'_h$), $\bound{y_1} = \bound{y_1'}$ and $\bound{y_2} = \bound{y_2'}$.
        From $\bound{x'} = h$, $x' R'_i y'_1$ and $x' R'_i y'_2$, Lemma \ref{lem:bound} gives $\bound{y'_1} \geq h-1$ and $\bound{y'_2} \geq h-1$ and, thus, $\bound{y_1} \geq h-1$ and $\bound{y_2} \geq h-1$.
        From the proof of Theorem \ref{th:k-bisim}, we then have $(M, y_1) \bisim_{h-1} (M', y'_1)$ and $(M, y_2) \bisim_{h-1} (M', y'_2)$.
        Since $(M', y'_1) \bisim_{h-1} (M'', y'')$ and $(M', y'_2) \bisim_{h-1} (M'', y'')$, we get $y'_1 \bisim_{h-1} y'_2$ and, thus, $y_1 \bisim_{h-1} y_2$.
        By Definition \ref{def:least-repr}, $\bound{x} = h$ and $y_1 \bisim_{h-1} y_2$ imply $\leastRepr{y_1}{h-1} = \leastRepr{y_2}{h-1}$.
        By Definition \ref{def:rooted-k-contr}, we then get $y'_1 = \reprClass{\leastRepr{y_1}{h-1}} = \reprClass{\leastRepr{y_2}{h-1}} = y'_2$, contradicting $y'_1 \neq y'_2$, and we're done.
    \end{proof}

    \section{Exponential Succinctness}\label{sec:succinctness}
    In this section, we show that, for any $k \geq 0$, rooted $k$-contractions can be exponentially more succinct than standard $k$-contractions. This means that we can create models of arbitrary size for which the rooted $k$-contraction is exponentially smaller than the corresponding standard $k$-contraction.   

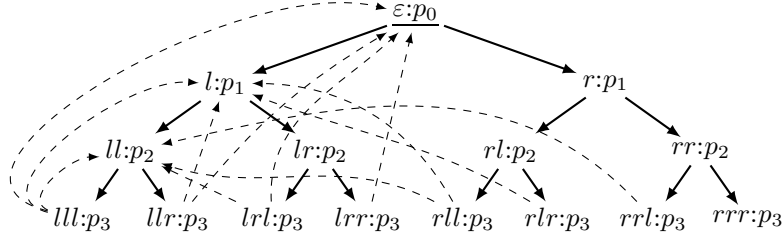
\begin{figure}
    \begin{center}
    \hspace{-1.2cm}
    \begin{tikzpicture}[level distance=0.9cm,
        level 1/.style={sibling distance=5cm},
        level 2/.style={sibling distance=2.5cm},
        level 3/.style={sibling distance=1.25cm},
        grow=down,
        edge from parent/.style={draw,thick,-latex},
        inner sep=2pt]
        \node (e) {\underline{$\varepsilon{:} p_0$}}
          child {node (l) {$l{:} p_1$}
            child {node (ll) {$ll{:} p_2$}
              child {node (lll) {$lll{:} p_3$}}
              child {node (llr) {$llr{:} p_3$}}
            }
            child {node (lr) {$lr{:} p_2$}
              child {node (lrl) {$lrl{:} p_3$}}
              child {node (lrr) {$lrr{:} p_3$}}
            }
          }
          child {node (r) {$r{:}p_1$}
            child {node (rl) {$rl{:}p_2$}
              child {node (rll) {$rll{:}p_3$}}
              child {node (rlr) {$rlr{:}p_3$}}
            }
            child {node (rr) {$rr{:}p_2$}
              child {node (rrl) {$rrl{:}p_3$}}
              child {node (rrr) {$rrr{:}p_3$}}
            }
          };
          \draw[-latex, dashed, bend left, out=135, in=135, looseness=1] (lll) to (e);
          \draw[-latex, dashed, bend left, out=125, in=135, looseness=1] (lll) to (l);
          \draw[-latex, dashed, bend left, out=110, in=135, looseness=1] (lll) to (ll);
          \draw[-latex, dashed, bend left, out=10, in=-190] (llr) to (e);
          \draw[-latex, dashed] (llr) to (l);
          \draw[-latex, dashed, bend left, out=40, in=-190] (lrl) to (e);
          \draw[-latex, dashed] (lrl) to (ll);
          \draw[-latex, dashed] (lrr) to (e);

          \draw[-latex, dashed, bend right] (rll) to (l);
          \draw[-latex, dashed, bend right, out=-20, in=170] (rll) to (ll);

          \draw[-latex, dashed, bend right, out=-10, in=180] (rlr) to (l);
        
          \draw[-latex, dashed,bend left,looseness=1,out=320,in=200] (rrl) to (ll);
      \end{tikzpicture}
       \end{center}    
        \caption{The pointed model $\mathcal{M}_k = ((W_k,R_k,V_k),\varepsilon)$ with designated world $\varepsilon$, for $k=3$. The solid and dashed edges are the accessibility edges of $\Box_s$ and $\Box_d$, respectively.}
        \label{fig:binary-tree}
    \end{figure}
    A binary tree of height $k$ has $2^{k+1}-1$ nodes. We will build a model $\mathcal{M}_k = ((W_k,R_k,V_k),\varepsilon)$ on a binary tree of height $k$ such that the standard $k$-contraction of $\mathcal{M}_k$ still has all $2^{k+1}-1$ nodes as worlds, but the rooted $k$-contraction is just a path in the tree, so a model with $k+1$ worlds. Hence the rooted $k$-contraction is exponentially smaller.  
    
    The nodes of a binary tree of height $k$ can be represented as strings of length at most $k$ over the alphabet $\{l,r\}$, where $l$ and $r$ represent the left and right child of a node, respectively. This is illustrated in Figure~\ref{fig:binary-tree} for $k=3$. The root is named $\varepsilon$ (the empty string), and the left and right children of a node $\sigma$ are $\sigma l$ and $\sigma r$, respectively. More precisely, we let the set of worlds $W_k$ of $\mathcal{M}_k$ be the set of strings of length at most $k$ over the alphabet $\{l,r\}$, \ie $W_k = \{ \sigma \in \{l,r\}^\ast \mid | \sigma | \leq k \}$. The tree edges (solid edges in Figure~\ref{fig:binary-tree}) are the edges from $\sigma$ to $\sigma l$ (left child) and to $\sigma r$ (right child). We let $\modalitiesSet = \{s, d\}$, where the accessibility relation $(R_k)_s$ of the modality $\Box_s$ represents the solid edges. Hence we let $(R_k)_s = \{ (\sigma,\sigma \alpha ) \in W_k \times W_k \mid \alpha \in \{l,r\} \}$. The accessibility relation $(R_k)_d$ of the modality $\Box_d$ represents the dashed edges, described later.
    
    As shown in Figure~\ref{fig:binary-tree}, the valuation function $V_k$ of $\mathcal{M}_k$ is such that each world at depth $n$ of the tree makes $p_n$ true and all other propositions false. More precisely, the set of atomic propositions of $\mathcal{M}_k$ is $\atomSet = \{p_0, \dots, p_k\}$, and we
    let $V_k(p_n) = \{ \sigma \in W_k \mid | \sigma | = n \}$. Suppose we decided to let $(R_k)_d = \emptyset$, \ie ignore the dashed edges. Then, the model is simply a binary tree where each world is labelled by an atomic proposition denoting the depth of the world. Clearly, any two worlds at the same depth are then bisimilar. Hence, the bisimulation contraction of $\mathcal{M}_k$ will simply be a chain model with $k+1$ worlds where the first world is labelled $p_0$, the second $p_1$, etc. 
    We add the dashed edges to ensure that $\mathcal{M}_k$ cannot be contracted further when considering standard $k$-contractions, but where the rooted $k$-contraction is still the simple chain model.   
    
    We now describe the dashed edges of the model. The dashed edges are from the leaf nodes of the binary tree to nodes of the leftmost branch of the tree. The crucial property of these edges is that each leaf node has edges to a different subset of the nodes of the leftmost branch. Hence no two leaf nodes will satisfy the same formulas of modal depth 1. Our specific choice is to put an edge from a leaf node $\sigma$ to the leftmost node at depth $n$ iff the $(n+1)$st letter of $\sigma$ is $l$, see Figure~\ref{fig:binary-tree}. More precisely, we let $(R_k)_d = \{ (\alpha_1 \cdots \alpha_k, l^n) \in W_k \times W_k \mid \alpha_{n+1} = l \}$ (where $l^n$ as usual denotes the string with $n$ occurrences of $l$). 

    \begin{lemma}
        Let $\sigma,\tau$ be distinct worlds of $\mathcal{M}_k$ at depth $n$. Then $\sigma \bisim_{k-n} \tau$ but not $\sigma \bisim_{k-n+1} \tau$.
    \end{lemma} 
    \begin{proof}
        Since $\sigma$ and $\tau$ are worlds at depth $n$ of the tree, they satisfy the same formulas up to modal depth $k-n$, since none of the dashed edges can be reached by formulas of modal depth $\leq k -n$ from worlds at depth $n$. From Proposition~\ref{prop:k-bisim} we can then conclude that $\sigma \bisim_{k-n} \tau$, as required. 

        Since $\sigma \neq \tau$, they differ in at least one of their positions, say position $m+1$ (\ie they differ in the $(m+1)$st letter). Suppose without loss of generality that $\sigma$ has $l$ in position $m+1$ and $\tau$ has $r$.  
        Then by construction of $\mathcal{M}_k$ we have that the formula $\Diamond_s^{n-k} \Diamond_d\ p_{m}$ is true in $\sigma$ but not in $\tau$: There exists a solid path of length $n-k$ from $\sigma$ to a leaf node that can see a $p_{m}$ world via a dashed edge, but there is no such path from $\tau$. Hence $\sigma$ and $\tau$ are not $(n-k+1)$-bisimilar. 
    \end{proof}       

    We can now reason as follows. Since the standard $k$-contraction only identifies worlds that are $k$-bisimilar, it will never be able to identify worlds at the same depth of the tree since, according to the lemma, any such two distinct worlds $\sigma$ and $\tau$ are not $(k-n+1)$-bisimilar and hence not $k$-bisimilar. 
    It will not be able to identify worlds at different levels either, as they have distinct valuations. Hence the standard $k$-contraction of $\mathcal{M}_k$ will contain all worlds of the original tree, so $2^{k+1}-1$ worlds. 
    Compare this to the rooted $k$-contraction. The worlds of the rooted $k$-contraction are of the form $\reprClass{\sigma}$ where $\sigma \in \maxRepr{W_k}$. Note that if $\sigma$ is a world of $\mathcal{M}_k$ at depth $n$, then $\bound{\sigma} = k - n$, and hence $\reprClass{\sigma} = \class{\sigma}{k-n}$. By the lemma, we then get that any two worlds $\sigma,\tau$ at the same depth of $\mathcal{M}_k$ belong to the same representative class $\reprClass{\sigma}$ (since they are $(k-n)$-bisimilar). Hence the rooted $k$-contraction can have at most $k+1$ worlds, one per level of the tree. We have hence proved the following succinctness result. 
    \begin{theorem}[Exponential succinctness] 
        There exist models $\mathcal{M}_k$, $k \geq 0$, for which the rooted $k$-contraction has $\Theta(k)$ worlds whereas the standard $k$-contraction has $\Theta(2^k)$ worlds.
    \end{theorem}

    \bibliographystyle{aiml}

\begin{thebibliography}{1}
        \expandafter\ifx\csname url\endcsname\relax
          \def\url#1{\texttt{#1}}\fi
        \expandafter\ifx\csname urlprefix\endcsname\relax\def\urlprefix{URL }\fi
        \newcommand{\enquote}[1]{``#1''}
        
        \bibitem{belle2022epistemic}
        Belle, V., T.~Bolander, A.~Herzig and B.~Nebel, \emph{Epistemic planning:
          Perspectives on the special issue}, Artificial Intelligence  (2022),
          p.~103842.
        
        \bibitem{book/cup/Blackburn2001}
        Blackburn, P., M.~d. Rijke and Y.~Venema, \enquote{Modal Logic,} Cambridge
          Tracts in Theoretical Computer Science, Cambridge University Press, 2001.
        
        \bibitem{blackburn2006modal}
        Blackburn, P. and J.~van Benthem, \emph{Modal logic: A semantic perspective},
          in: \emph{Handbook of Modal Logic} (2006).
        
        \bibitem{bolander2011epistemic}
        Bolander, T. and M.~B. Andersen, \emph{Epistemic planning for single- and
          multi-agent systems}, Journal of Applied Non-Classical Logics \textbf{21}
          (2011), pp.~9--34.
        
        \bibitem{bolander2020del}
        Bolander, T., T.~Charrier, S.~Pinchinat and F.~Schwarzentruber,
          \emph{{DEL}-based epistemic planning: Decidability and complexity},
          Artificial Intelligence \textbf{287} (2020), pp.~1--34.
        
        \bibitem{journals/logcom/2023/BolanderL}
        Bolander, T. and A.~Lequen, \emph{Parameterized complexity of dynamic belief
          updates: {A} complete map}, J. Log. Comput. \textbf{33} (2023),
          pp.~1270--1300.
        
        \bibitem{books/el/07/GorankoO}
        Goranko, V. and M.~Otto, \emph{Model theory of modal logic}, in: P.~Blackburn,
          J.~{Van Benthem} and F.~Wolter, editors, \emph{Handbook of Modal Logic},
          Studies in Logic and Practical Reasoning  \textbf{3}, Elsevier, 2007 pp.
          249--329.
        
        \bibitem{conf/ijcai/Yu2013}
        Yu, Q., X.~Wen and Y.~Liu, \emph{Multi-agent epistemic explanatory diagnosis
          via reasoning about actions}, in: F.~Rossi, editor, \emph{{IJCAI} 2013,
          Proceedings of the 23rd International Joint Conference on Artificial
          Intelligence, Beijing, China, August 3-9, 2013} (2013), pp. 1183--1190.
        
    \end{thebibliography}

\end{document}